\documentclass[envcountsame]{llncs}
    \usepackage{amsmath}
    \usepackage{amsxtra}
    \usepackage{amstext}
    \usepackage{amssymb}
    \usepackage{latexsym}
    \usepackage{dsfont} 
\newcommand{\tr}[2][1]{Tr_{#1}^{#2}}

\newcommand{\GF}[1]{{\mathbb F}_{#1}}

\begin{document}
\title{Preimages of $p-$Linearized Polynomials over $\GF{p}$}
\author{Kwang Ho Kim\inst{1,2}\and Sihem Mesnager\inst{3} \and Jong Hyok Choe\inst{1} \and Dok Nam Lee\inst{1} }
\institute{ Institute of Mathematics, State Academy of Sciences,
Pyongyang, Democratic People's Republic of Korea\\
\email{khk.cryptech@gmail.com} \and PGItech Corp., Pyongyang, Democratic People's Republic of Korea\\ \and Department of Mathematics, University of Paris VIII, F-93526 Saint-Denis, Laboratory Geometry, Analysis and Applications,  LAGA, University Sorbonne Paris Nord, CNRS, UMR 7539,  F-93430, Villetaneuse, France, and Telecom
Paris,  91120 Palaiseau, France.\\
\email{smesnager@univ-paris8.fr}\\} \maketitle

\begin{abstract}
Linearized polynomials over finite fields have been intensively
studied over the last several decades. Interesting new applications
of linearized polynomials to coding theory and finite geometry have
been also highlighted in recent years.

Let $p$ be any prime. Recently, preimages of the $p-$linearized
polynomials $\sum_{i=0}^{\frac kl-1} X^{p^{li}}$ and
$\sum_{i=0}^{\frac kl-1} (-1)^i X^{p^{li}}$ were explicitly computed
over $\GF{p^n}$ for any $n$. This paper extends that study to
$p-$linearized polynomials over $\GF{p}$, i.e., polynomials of the
shape
$$L(X)=\sum_{i=0}^t \alpha_i X^{p^i}, \alpha_i\in\GF{p}.$$ Given a
$k$ such that $L(X)$ divides $X-X^{p^k}$, the preimages of $L(X)$
can be explicitly computed over $\GF{p^n}$ for any $n$.\\


\end{abstract}

\section{Introduction}
Let $p$ be a prime. A polynomial $L(X)\in\GF{p}[X]$ of shape
\begin{equation}\label{p-poly}
  L(X) = \sum_{i=0}^t \alpha_i X^{p^i}, \alpha_i\in\GF{p}
\end{equation}
is called a $p-$linearized polynomial over $\GF{p}$ or simply a
$p-$polynomial. Let $n$ be a positive integer. An affine equation
over $\GF{p}$ is an equation of type
\begin{equation}\label{eq_1}
L(X)=a,
\end{equation}
where $L$ is a $p-$polynomial and $a\in \GF{p^n}$.

The goal of this study is to explicitly compute all the solutions in
$\GF{p^n}$ to Equation~\eqref{eq_1}. We can reduce our study to the
$p-$polynomials with $\alpha_0\neq 0$ because if
$\alpha_0=\alpha_1=\cdots =\alpha_{s-1}=0, \alpha_s\neq 0$ then
$L(X)=\left(\sum_{i=0}^{t-s} \alpha_{i+s} X^{p^i}\right)^{p^s}$ and
therefore we can instead consider the affine equation
$L'(X)=a^{\frac{1}{p^s}}$ the first order term of which has the
nonzero coefficient $\alpha_s$.

Affine equations arise in many different problems and contexts (e.g.
\cite{BSS1999,CarletBook,CarletBook1,C20,C2019109,KCM20,KM20,MCGUIRE201968,McGuireMueller,MKCL2019,MKCT20,MKJ20,PF19,WU201379,Z19}).
 However, to find explicit
solutions  is often challenging and it is the ultimate goal to
achieve.

Recently, preimages of the special $p-$linearized polynomials
$T_l^k(X)=\sum_{i=0}^{\frac kl-1} X^{p^{li}}$ and
$S_l^k(X)=\sum_{i=0}^{\frac kl-1} (-1)^i X^{p^{li}}$ were explicitly
computed over $\GF{p^n}$ for any $n$, for the specific case $p=2$ in
\cite{MKCL2019} and for any $p$ in \cite{MKCL2020}.

This paper extends that studies to any $p-$polynomials of the shape
\eqref{p-poly}. Given a $k$ such that $L(X)$ divides $X-X^{p^k}$,
the preimages of $L(X)$ can be explicitly computed over $\GF{p^n}$.
It is known that $T_l^k(X)$ and $S_l^k(X)$ divide $X-X^{p^{2k}}$.

Obviously, if $x_1$ and $x_2$ are two solutions in $\GF{p^n}$ to
Equation~\eqref{eq_1}, then their difference $x_1-x_2$ is a zero of
$L$ in $\GF{p^n}$, that is, their difference lies in the set
$\ker(L)\cap \GF{p^n}:=\{x\in\GF{p^n}\mid L(x) = 0\}$. Therefore,
determination of the $\GF{p^n}-$solutions to Equation~\eqref{eq_1}
can be divided into two problems: to determine $\ker(L)\cap
\GF{p^n}$ and to find an explicit solution $x_0$ in $\GF{p^n}$.

The paper is organized as follows. After introducing some prerequisites in Section \ref{sec2}, we
solve these two problems for $p-$polynomials $L(X)$ in Section \ref{Main-results}.  As a
by-product, we also characterize the elements $a$ in $\GF{p^n}$ for
which Equation~\eqref{eq_1} has at least one solution in $\GF{p^n}$.
In Section \ref{Examples}, we provide some pertinent examples which should explain
applicability of the obtained results. Section \ref{Conclusion} concludes the article.

\section{Some prerequisites}\label{sec2}
Given a finite set $F$, $|F|$ denotes its cardinality.
 For two polynomials $f,g\in\GF{p}[X]$, we denote
their greatest common divisor monic polynomial in $\GF{p}[X]$ as
$(f,g)$.

Define linearized polynomials:
\begin{eqnarray}
  \label{equation:T}  &&T_l^k(X) := \sum_{i=0}^{\frac kl-1} X^{p^{li}},\\
  \label{equation:S}  &&S_l^k(X) := \sum_{i=0}^{\frac kl-1} (-1)^i  X^{p^{li}}.
\end{eqnarray}
Notice $$S_k^{2k}(X)=X-X^{p^k},$$ and therefore, $x\in \GF{p^k}$ if
and only if $S_k^{2k}(x)=0$.

 The following properties of these
polynomials are heavily used throughout this paper.
\begin{lemma}[\cite{MKCL2019}]\label{lem_properties}
  For any positive integers $k$, $l$ and $m$ with $m\vert l\vert k$,
  the following are true.
\begin{enumerate}
\item \label{lem_properties:1}\(T_l^k\circ T_m^l(X)=T_m^k(X)\) is an identity. \(T_l^k\circ S_m^l(X)=S_m^k(X)\) if $l/m$ is even and
  \( S_l^k\circ S_m^l(X)=S_m^k(X)\) if $l/m$ is odd.
\item\label{lem_properties:2} \(S_l^k\circ T_l^{2l}(X)=S_k^{2k}(X)\) if
  $\frac{k}{l}$ is even and   \(S_l^k\circ T_l^{2l}(X)=T_k^{2k}(X)\) if $\frac{k}{l}$ is odd.
\item \label{lem_properties:3}  \(T_l^k\circ S_l^{2l}(X) =S_k^{2k}(X)\).
\item \label{lem_properties:4} \( T_k^{[n,k]}(x)=T_d^n(x)\) for
  any $x\in\GF{p^n}$. Furthermore, if
  $\frac{[n,k]}{k}$ is even, then \( S_k^{[n,k]}(x)=S_d^n(x)\) for
  any $x\in\GF{p^n}$.
\end{enumerate}
\end{lemma}

The  set of $p-$polynomials over $\GF{p}$ forms an integral domain
under the operations of symbolic multiplication (composition of
polynomials) and ordinary addition (e.g. see page 115 in
\cite{Lidl1997}). Therefore, under the symbolic multiplication, any
two $p-$polynomials over $\GF{p}$ are commutative. This fact will be
implicitly used throughout this paper.

\begin{definition} [Definition 3.58 in \cite{Lidl1997}]
 The polynomials
\[
l(X)=\sum_{i=0}^{t}\alpha_i x^i \text{ and }
L(X)=\sum_{i=0}^t\alpha_i x^{p^i}
\]
over $\GF{p}$ are called $p-$associates of each other. More
specifically, $l(X)$ is the conventional $p-$associate of $L(X)$ and
$L(X)$ is the linearized $p-$associate of $l(X)$.
\end{definition}

In this paper, linearized polynomials are capitalized and the
conventional $p-$associates are denoted by the corresponding
lowercases. For example, we denote the conventional $p-$associates
of $T_l^k(X),S_l^k(X)$ as $t_l^k(X),s_l^k(X)$ respectively. For the
linearized $p-$associate of a polynomial $l(X)$, we sometimes also
use the denotation $\overline{l}(X)$ (together with the promised
denotation $L(X)$).

Let $L_1(X)$ and $L_2(X)$ be $p-$polynomials over $\GF{p}$ with
conventional $p-$associates $l_1(X)$ and $l_2(X)$. Then,
$l_1(X)+l_2(X)$ and $L_1(X)+L_2(X)$ are $p-$associates of each
other, and $l_1(X)\cdot l_2(X)$ and $L_1\circ L_2(X)$ are
$p-$associates of each other (e.g. see Lemma 3.59 in
\cite{Lidl1997}). Furthermore, the following properties  are
equivalent:  (i) $L_1(X)$ symbolically  divides $L_2(X)$; (ii)
$L_1(X)$ divides $L_2(X)$ in the ordinary sense;  (iii) $l_1(X)$
divides $l_2(X)$ (Theorem 3.62 in \cite{Lidl1997}). From these
facts, one deduces
\begin{proposition}\label{interKer}
Let $(l_1,l_2)=l_0$ and $L_1,L_2,L_0$ be $p-$polynomials over
$\GF{p}$ with conventional $p-$associates $l_1,l_2,l_0$,
respectively. Then
$$\ker(L_0)=\ker(L_1)\cap\ker(L_2).$$
\end{proposition}
\begin{proof} There exist $l'_1, l'_2$ such that
$l_1l'_1+l_2l'_2=l_0.$ Then, it holds
$$L_1\circ L'_1+L_2\circ L'_2=L_0.$$ Therefore, $\ker(L_1)\cap\ker(L_2)\subset \ker(L_0).$
Since $l_0$ divides $l_1$ and $l_2$, also $L_0$ divides $L_1$ and
$L_2$. Hence, $\ker(L_1)\cap\ker(L_2)\supset \ker(L_0).$ This proves
the equality.\qed
\end{proof}
\section{Main results}\label{Main-results}

As mentioned in the introductory section, we can restrict our study
to $p-$polynomials without multiple roots. Let $L(X)$ be a
$p-$polynomial without multiple roots. Further, let us suppose that
some positive integer $k$ is given such that the $p-$polynomial
$L(X)$ symbolically divides $S_k^{2k}(X)$ (such a $k$ always
exists; for example, we can think of the degree of the splitting
field of $L$). Let $d=\gcd(n,k)$. Example 3.61 of \cite{Lidl1997}
shows that then there exists a $p-$polynomial $L'(X)$ over $\GF{p}$
such that
\[
S_k^{2k}(X)=L\circ L'(X)=L'\circ L(X)=L'(L(X)).
\]

The following fact is simple but very useful.

\begin{proposition}\label{KerL}
Let $L$ and $L'$ be $p-$polynomials. If $S_k^{2k}(X)=L\circ L'(X)$
for some positive integer $k$, then $\ker(L)=L'(\GF{p^k})$.
\end{proposition}
\begin{proof}
In fact, it is obvious $L'(\GF{p^k})\subset \ker(L)$, and on the
other hand $|L'(\GF{p^k})|=\frac{p^k}{|\ker(L')\cap \GF{p^k}|}\geq
\frac{p^k}{|\ker(L')|}= \frac{p^k}{\deg(L')}=\deg(L)=|\ker(L)|$, and
therefore $\ker(L)=L'(\GF{p^k})$.\qed
\end{proof}

  To start our study, we should define some polynomials which will be
needed to establish our main results. Define $w(X)=(l'(X),
t_d^k(X)),$ $u(X):=\frac{l'(X)}{w(X)}$ and
$v(X):=\frac{t_{d}^k(X)}{w(X)}.$ Since $(u(X),v(X))=1$, we can
define $f=u^{-1} \mod v$ and $g=v^{-1} \mod u$ so that $u\cdot
f+v\cdot g=1.$ Then, by the properties of $p-$associates mentioned
in Sec.~\ref{sec2}, obviously it holds
\begin{equation}\label{invV}
 U\circ
F(X)+V\circ G(X)=X,
\end{equation}
\begin{equation}\label{UWL'VWT}
T_d^k=V\circ W \text{ and } U\circ W=L' \text{ and } U\circ T_d^k=L'
\circ V.
\end{equation}

The following fact will be useful in the sequel.
\begin{proposition}\label{LU}
Using the notation above, one has
$$U\circ L=S_d^{2d}\circ V.$$
\end{proposition}
\begin{proof} In fact,
\begin{align*}
U\circ L&=\overline{l}\circ \overline{\left(\frac{ l'}{w}\right)}=
\overline{\left(\frac{l\cdot
l'}{w}\right)}=\overline{\left(\frac{s_k^{2k}}{w}\right)}
\\
&=\overline{\left(\frac{s_d^{2d}\cdot t_d^k}{w}\right)} \text{ (by
Item~\ref{lem_properties:3} of
Lemma~\ref{lem_properties})}\\
&=\overline{s_d^{2d}}\circ
\overline{\left(\frac{t_d^k}{w}\right)}=S_d^{2d}\circ V.
\end{align*}\qed
\end{proof}

\subsection{Determination of $\ker(L)\cap \GF{p^n}$} We first determine the kernel
of $L$ in $\GF{p^n}$.
\begin{lemma}\label{KerinFF}
$\ker(L)\cap \GF{p^n}=U(\GF{p^{d}}).$
\end{lemma}
\begin{proof} Both sides of the equality to be proved can be
rewritten as:
\begin{align*}\label{eq_4} \ker(L)\cap
\GF{p^n}&\overset{\text{Prop.~\ref{KerL}}}{=}L'(\GF{p^k})\cap\GF{p^n}=
L'(\GF{p^k})\cap\GF{p^d}\\&=L'\left(\{\beta\in\GF{p^k}\mid
S_{d}^{2d}\circ L'(\beta)=0\}\right).
\end{align*}
and
\begin{equation}\label{UsubsetKerL}
U(\GF{p^{d}})=U\left(T_{d}^k(\GF{p^k})\right)\overset{\eqref{UWL'VWT}}{=}L'\left(V(\GF{p^k})\right).
\end{equation}
Thus it is enough for our purpose to prove
\[
\{\beta\in\GF{p^k}\mid S_{d}^{2{d}}\circ L'(\beta)=0\}=V(\GF{p^k}).
\]
By using \eqref{UsubsetKerL}, it can be easily checked
\begin{equation}\label{VsubsetKerL0}
\{\beta\in\GF{p^k}\mid S_{d}^{2{d}}\circ L'(\beta)=0\}\supset
V(\GF{p^k}).
\end{equation}
From $T_d^k(X)=V\circ W(X)$ (see \eqref{UWL'VWT}), it is obvious
that $W(X)$ has no multiple root and so $$\deg(W)=|\ker(W)|.$$
Furthermore,
$\deg(V)=\frac{\deg(T_d^k)}{\deg(W)}=\frac{p^{k-d}}{\deg(W)}$ and
\begin{equation}\label{eq_7}
|V(\GF{p^k})|\geq
\frac{p^k}{\deg(V)}=\frac{p^k}{\frac{p^{k-{d}}}{\deg(W)}}=p^{d}\cdot
\deg(W).
\end{equation}
From Item~\ref{lem_properties:3} of Lemma~\ref{lem_properties}, one
has $\ker(T_d^k)=S_{d}^{2{d}}(\GF{p^k})$ and thus
\begin{align*}
|\{\beta\in\GF{p^k}\mid S_{d}^{2{d}}\circ L'(\beta)=0\}|&=p^{d}\cdot
|\{\beta\in S_{d}^{2{d}}(\GF{p^k})\mid L'(\beta)=0\}|\\
&=p^{d}\cdot |\{\beta\in \ker(T_d^k)\mid L'(\beta)=0\}|\\
&=p^{d}\cdot |\ker(T_d^k)\cap \ker(L')|\\
&\overset{\text{Prop.~\ref{interKer}}}{=}p^{d}\cdot
|\ker(W)|=p^{d}\cdot\deg(W).
\end{align*}
That is,
\begin{equation}\label{VsupsetKerL0}
|V(\GF{p^k})|\geq |\{\beta\in\GF{p^k}\mid S_{d}^{2{d}}\circ
L'(\beta)=0\}|.
\end{equation}
Combining \eqref{VsubsetKerL0} and \eqref{VsupsetKerL0} completes
the proof. \qed
\end{proof}

\subsection{A solution to the equation $L(X)=a\in \GF{p^n}$}
Now, reminding the definitions of polynomials $U,W,G$ in the
beginning of this section, we provide an explicit expression for a
solution to Equation~\eqref{eq_1}.
\begin{theorem}\label{MainTheorem}  Let $\delta\in \GF{p^n}^*$ and $\delta_1\in \GF{p^{2d}}^*$ be any elements such that
$T_{d}^n(\delta)=1$ and  $S_{d}^{2d}(\delta_1)=1$.
Equation~\eqref{eq_1} has a solution in $\GF{p^n}$ if and only if
$$U\circ T_d^n(a)=0.$$
Specifically, when $U\circ T_d^n(a)=0$,
\begin{equation}\label{particular_solution}
x_0=W(y)+U\circ G(\delta_1\cdot (a-L\circ W(y))),
\end{equation}
where
$$y=\sum_{i=0}^{\frac{n}{d}-2}\sum_{j=i+1}^{\frac{n}{d}-1}\delta^{p^{kj}}U(a)^{p^{ki}},$$
is a particular solution in $\GF{p^n}$ to Equation~\eqref{eq_1}.
\end{theorem}
\begin{proof}
$(\Longrightarrow)$ Let $L(x_0)=a$ for $x_0\in\GF{p^n}$. Then
\begin{align*}
U\circ T_d^n(a)&= U\circ T_d^n\circ L(x_0)=U\circ L\circ T_d^n(x_0)\\
&=S_{d}^{2d}\circ
V\circ T_d^n(x_0) \text{ (by Proposition~\ref{LU})}\\
&=V\circ S_n^{2n}(x_0) \text{ (by Item~\ref{lem_properties:3} of Lemma~\ref{lem_properties})}\\
&=V(0)=0.
\end{align*}

$(\Longleftarrow)$ Now,  we will show that under assumption $U\circ
T_d^n(a)=0$, \eqref{particular_solution} is really a particular
solution in $\GF{p^n}$ to equation (\ref{eq_1}). First of all, we
have
\begin{align*}
L\circ L'(y)&=y-y^{p^k}\\
&=\sum_{i=0}^{\frac{n}{d}-2}\sum_{j=i+1}^{\frac{n}{d}-1}\delta^{p^{kj}}U(a)^{p^{ki}}-\sum_{i=0}^{\frac{n}{d}-2}\sum_{j=i+1}^{\frac{n}{d}-1}\delta^{p^{k(j+1)}}U(a)^{p^{k(i+1)}}\\
&=\sum_{i=0}^{\frac{n}{d}-2}\sum_{j=i+1}^{\frac{n}{d}-1}\delta^{p^{kj}}U(a)^{p^{ki}}-\sum_{i=1}^{\frac{n}{d}-1}\sum_{j=i+1}^{\frac{n}{d}}\delta^{p^{kj}}U(a)^{p^{ki}}\\
&=\sum_{j=1}^{\frac{n}{d}-1}\delta^{p^{kj}}U(a)-\sum_{i=1}^{\frac{n}{d}-2}\delta^{p^{\frac{nk}{d}}}U(a)^{p^{ki}}-\delta^{p^{\frac{nk}{d}}}U(a)^{p^{k(\frac{n}{d}-1)}}\\
&=\sum_{j=1}^{\frac{n}{d}-1}\delta^{p^{kj}}U(a)-\sum_{i=1}^{\frac{n}{d}-1}\delta U(a)^{p^{ki}}\\
&=\sum_{j=1}^{\frac{n}{d}-1}\delta^{p^{kj}}U(a)-\delta\sum_{i=1}^{\frac{n}{d}-1}U(a)^{p^{ki}}\\
&=\sum_{j=0}^{\frac{n}{d}-1}\delta^{p^{kj}}U(a)-\delta\sum_{i=0}^{\frac{n}{d}-1}U(a)^{p^{ki}}\\
&=T_d^n(\delta)U(a)-\delta T_k^{[n,k]}\circ U(a)\\
&=U(a)-\delta T_d^{n}\circ U(a) \text{ (by
Item~\ref{lem_properties:4} of
Lemma~\ref{lem_properties})}\\
&=U(a).
\end{align*}
That is,
$$L\circ L'(y)=U(a),$$
or, by using \eqref{UWL'VWT}, equivalently
\begin{equation}\label{kerUe}
U(a-L\circ W(y))=0.
\end{equation}

Set
\[
z=U\circ G(\delta_1\cdot (a-L\circ W(y))).
\]
Then, $z\in\GF{p^{d}}$ because
\begin{align*}
S_{d}^{2{d}}(z)&=S_{d}^{2{d}} \circ U\circ
G(\delta_1\cdot (a-L\circ W(y)))\\
&= G\circ U(a-L\circ W(y))\overset{\eqref{kerUe}}{=}0.
\end{align*}
 Furthermore,
\begin{align*}
L(z)&=L\circ U\circ G(\delta_1\cdot
(a-L\circ W(y)))\\
&\overset{\text{Prop.~\ref{LU}}}{=}S_{d}^{2{d}}\circ V\circ G(\delta_1\cdot (a-L\circ W(y)))\\
&=V\circ G\circ
S_{d}^{2{d}}(\delta_1\cdot (a-L\circ W(y)))\\
&=V\circ G(a-L\circ W(y))\overset{\eqref{invV}}{=}a-L\circ W(y)
\end{align*}
and therefore, with $x_0=W(y)+z \in\GF{p^n}$, one has
\[
L(x_0)=L\circ W(y)+L(z)=L\circ W(y)+a-L\circ W(y)=a.
\]
\qed
\end{proof}

If $p\nmid\frac{k}{d}$, a more concise formula for a particular
solution can be stated as follows.
\begin{theorem}\label{Maintheorem_1}
Suppose $p\nmid\frac{k}{d}$ and let  $\delta$ be any element of
$\GF{p^{2n}}$ such that $S_n^{2n}(\delta)=1$. Equation~\eqref{eq_1}
has a solution in $\GF{p^n}$ if and only if $$L'\circ T_d^n(a)=0.$$
Specifically,
\begin{equation}\label{particular_solution_1}
x_0=\frac{d}{k}\cdot L'\circ T_k^{[n, k]}(\delta\cdot a)
\end{equation}
is a solution to Equation~\eqref{eq_1} which belongs to $\GF{p^n}$
under the condition $L'\circ T_d^n(a)=0$.
\end{theorem}

\begin{proof} First of all,
\begin{align*}
L(x_0)&=\frac{d}{k}\cdot L\circ L'\circ T_k^{[n, k]}(\delta\cdot
a)=\frac{d}{k}\cdot S_k^{2k}\circ
T_k^{[n, k]}(\delta\cdot a)\\
&=\frac{d}{k}\cdot S_{[n, k]}^{2[n, k]}(\delta\cdot a)  \text{ (by
Item~\ref{lem_properties:3} of Lemma~\ref{lem_properties})}\\
&=\frac{d}{k}\cdot S_{[n, k]}^{2[n, k]}(\delta)\cdot a\\
&=\frac{d}{k}\cdot \frac{k}{d}\cdot S_{n}^{2n}(\delta)\cdot a=a.
\end{align*}
$(\Longrightarrow)$ Let $L(x)=a$ for $x\in\GF{p^n}$. Then,
\begin{align*}
L'\circ T_d^n(a)&=T_d^{n}\circ L'(a)\\
&=T_k^{[n,k]}\circ L'(a) \text{ (by Item~\ref{lem_properties:4} of Lemma~\ref{lem_properties})}\\
&=T_k^{[n,k]}\circ L\circ L'(x)=T_k^{[n,k]}\circ
S_k^{2k}(x)\\
&=S_{[n, k]}^{2[n, k]}(x)   \text{ (by Item~\ref{lem_properties:3} of Lemma~\ref{lem_properties})}\\
&=0.
\end{align*}
$(\Longleftarrow)$

Assume $L'\circ T_d^n(a)=0$ i.e. $L'\circ T_k^{[n,k]}(a)=0$ (by
Item~\ref{lem_properties:4} of Lemma~\ref{lem_properties}).
 Then,
\begin{align*}
S_n^{2n}(x_0)=\frac{d}{k}\cdot L'\circ T_k^{[n, k]}\circ
S_n^{2n}(\delta\cdot a)=\frac{d}{k}\cdot L'\circ T_k^{[n, k]}(a)=0
\end{align*}
and therefore $x_0\in\GF{p^n}$. \qed
\end{proof}

\section{Examples}\label{Examples}
The following two examples show that this paper can be considered as
a generalization of \cite{MKCL2019} and \cite{MKCL2020}. Set $e
:=\gcd(n,l)=\gcd(d,l)$.

\begin{example}\label{ExamTkl}
 Let $L=T_l^k$. With $L'=S_l^{2l}$, Item~\ref{lem_properties:3} of Lemma~\ref{lem_properties} let us know $L\circ
L'=S_k^{2k}$.  In this case,
$$w(X)=(t_d^k(X),
s_l^{2l}(X))=\left(\frac{1-X^k}{1-X^d},
1-X^l\right)=\begin{cases}1-X^l, &\text{
if $p\mid \frac{k}{[d,l]}$}\\
\frac{1-X^l}{1-X^e}, &\text{ otherwise.} \end{cases}$$ Hence,
$$u(X)=\frac{1-X^l}{w(X)}=\begin{cases}1, &\text{
if $p\mid \frac{k}{[d,l]}$}\\
1-X^e, &\text{ otherwise.} \end{cases}$$ and
$$U(X)=\overline{u}(X)=\begin{cases}X, &\text{
if $p\mid \frac{k}{[d,l]}$}\\
S_e^{2e}(X), &\text{ otherwise.} \end{cases}$$ Therefore,
Lemma~\ref{KerinFF} gives Theorem 1 in \cite{MKCL2020} as a
consequence. Also, Theorem~\ref{MainTheorem} of this paper and
Theorem 4 in \cite{MKCL2020} provide the same ``if and only if''
conditions for $T_l^k(X)=a, a\in \GF{p^n}$ to has a solution in
$\GF{p^n}$, which is $U\circ T_d^n(a)=0$. When $p\mid
\frac{k}{[d,l]}$, it is also evident that the explicit expressions
for a particular solution also coincide in both papers. However,
when $p\nmid \frac{k}{[d,l]}$, it is not obvious whether or not the
particular solutions in both papers coincide too.\qed
\end{example}

\begin{example}\label{ExamSkl}
 Let $L=S_l^k$. Item~\ref{lem_properties:2} of Lemma~\ref{lem_properties} let us know that $L\circ
L'=S_k^{2k}$ with $L'=T_l^{2l}$ if $\frac{k}{l}$ is even and
 that $L\circ
L'=S_{2k}^{4k}$ with $L'=T_l^{2l}\circ S_k^{2k}$ if $\frac{k}{l}$ is
odd.

If $\frac{k}{l}$ is even, then
$$w(X)=\left(\frac{1-X^k}{1-X^d}, 1+X^l\right)=\begin{cases}1+X^l,
&\text{
if $\frac{d}{e}$ is odd, or, $\frac{d}{e}$ is even and $p\mid \frac{k}{[d,l]}$}\\
\frac{1+X^l}{1+X^e}, &\text{ if $\frac{d}{e}$ is even and $p\nmid
\frac{k}{[d,l]}$.} \end{cases}$$

Hence,
$$u(X)=\frac{1+X^l}{w(X)}=\begin{cases}1, &\text{
if $\frac{d}{e}$ is odd, or, $\frac{d}{e}$ is even and $p\mid \frac{k}{[d,l]}$}\\
1+X^e, &\text{ if $\frac{d}{e}$ is even and $p\nmid
\frac{k}{[d,l]}$.} \end{cases}$$ and
$$U(X)=\overline{u}(X)=\begin{cases}X, &\text{
if $\frac{d}{e}$ is odd, or, $\frac{d}{e}$ is even and $p\mid \frac{k}{[d,l]}$}\\
T_e^{2e}(X), &\text{ if $\frac{d}{e}$ is even and $p\nmid
\frac{k}{[d,l]}$.} \end{cases}$$

Thus, Lemma~\ref{KerinFF} gives Theorem 2 in \cite{MKCL2020} as a
consequence, and  Theorem~\ref{MainTheorem} of this paper  provides
the same ``if and only if'' conditions for $S_l^k(X)=a, a\in
\GF{p^n}$ to has a solution in $\GF{p^n}$ as in Theorem 5 in
\cite{MKCL2020}.

If $\frac{k}{l}$ is odd, then
\begin{align*}w(X)&=\left(\frac{1-X^{2k}}{1-X^{(n,2k)}},
(1+X^l)(1-X^k)\right)\\&=\begin{cases}\frac{(1+X^l)(1-X^k)}{1-X^d},
&\text{
if $\frac{n}{d}$ is odd, or, $\frac{n}{d}$ is even and $p\mid \frac{k}{[d,l]}$,}\\
\frac{(1+X^l)(1-X^k)}{(1+X^e)(1-X^d)}, &\text{ if $\frac{n}{d}$ is
even and $p\nmid \frac{k}{[d,l]}$.} \end{cases}
\end{align*}
Hence,
$$U(X)=\begin{cases}S_d^{2d}, &\text{
if $\frac{n}{d}$ is odd, or, $\frac{n}{d}$ is even and $p\mid \frac{k}{[d,l]}$,}\\
T_e^{2e}\circ S_d^{2d}(X), &\text{ if $\frac{n}{d}$ is even and
$p\nmid \frac{k}{[d,l]}$.} \end{cases}$$

Since $S_d^{2d}\circ T_{2d}^n=S_d^n$ (Item~\ref{lem_properties:1} of
Lemma~\ref{lem_properties}) and $S_d^{2d}\circ T_d^n(\GF{p^n})=0$,
Lemma~\ref{KerinFF} gives Theorem 3 in \cite{MKCL2020} as a
consequence, and Theorem~\ref{MainTheorem} of this paper  provides
the same ``if and only if'' conditions for $S_l^k(X)=a, a\in
\GF{p^n}$ to has a solution in $\GF{p^n}$ as in Theorem 6 in
\cite{MKCL2020}.\qed
\end{example}

The following example discusses one case which could not be solved
in  \cite{MKCL2020} and \cite{MKCL2019}.

\begin{example} Let $L(X)=X+X^2+X^{2^3}\in \GF{2}[X].$ One wants to solve the equation
\begin{equation}\label{examEq2}
L(X)=a, a\in \GF{2^n}.
\end{equation} It holds $L\circ L'(X)=X^{2^7}+X$ with
$L'(X)=X+X^2+X^{2^2}+X^{2^4}$ and therefore one can take $k=7$. If
$7\mid n$, then $d=7$, and else $d=1$. Therefore, in this case
$$w(X)=\begin{cases}1, &\text{ if $7\mid n$,}\\
\left(\frac{1+X^7}{1+X}, 1+X+X^2+X^4\right)=1+X^2+X^3, &\text{
otherwise}\end{cases}$$ and
$$U(X)=\begin{cases}X+X^2+X^{2^2}+X^{2^4}, &\text{ if $7\mid n$,}\\
X+X^{2}, &\text{ otherwise.}\end{cases}$$ Lemma~\ref{KerinFF} gives:
$$\ker(L)\cap\GF{2^n}=\begin{cases}U(\GF{2^7})=\{x+x^2+x^{2^2}+x^{2^4}\mid x\in \GF{2^7}\}, &\text{ if $7\mid n$,}\\
U(\GF{2})=\{0\}, &\text{ otherwise.}\end{cases}$$
Theorem~\ref{MainTheorem} gives: If $7\nmid n$, then
Equation~\eqref{examEq2} has always a unique solution in $\GF{2^n}$
for any $a\in \GF{2^n}$ (since $U\circ T_d^n(a)=
\tr{n}(a)+\tr{n}(a)^2=0$); If $7\mid n$, then
Equation~\eqref{examEq2} has a solution in $\GF{2^n}$ if and only if
$T_7^n(a+a^2+a^{2^2}+a^{2^4})=0$. Theorem~\ref{MainTheorem} (and
Theorem~\ref{Maintheorem_1}) gives an explicit expression for the
rational solutions as well.\qed
\end{example}

\begin{remark}We finalize this paper by discussing the problem of how to
find such a $k$ that $L(X)$ divides $X-X^{p^k}$ for general
$p-$polynomial $L(X)$. It is obvious that one can consider any
multiple of the order (also called the exponent or the period) of
the conventional $p-$associate $l(X)$. There are many families of
special polynomials with known order. However, computing the order
of a general polynomial is equivalent to factoring that polynomial.
There have been developed various deterministic/random algorithms to
factor a given polynomial (e.g. such as Berlekamp $Q-$matrix method,
Camion's algorithm \cite{camion1983}, Cantor-Zassenhaus's algorithm
\cite{CZ1981}). Shoup's algorithm \cite{Shoup1990} completely
factors a polynomial of degree $N$ over $\GF{p}$ in $O(p^{1/2}\cdot
\log p \cdot N^{2+\epsilon })$ \cite{Shparlinski1992}. The  fastest
known randomized  algorithm for factorization in $\GF{p}[X]$ is the
Kaltofen-Shoup  algorithm \cite{KaltofenShoup1998} implemented by
Kedlaya-Umans fast modular composition \cite{KedlayaUmans08}. It
belongs in the Cantor-Zassenhaus \cite{CZ1981} framework and to
factor a polynomial of degree $N$ takes $O \left(N^{3/2+o(1)}(\log
p)^{1+o(1)}+N^{1+o(1)}(\log p)^{2+o(1)}\right)$ expected time.\qed
\end{remark}

\section{Conclusions}\label{Conclusion}

Linearized polynomials over finite fields are fundamental objects
having a lot of applications (in particular related to coding theory
and finite geometry) as highlighted by Gary McGuire in his invited
talk at the recent  International Workshop on the Arithmetic of
Finite Fields  (WAIFI 2020). In this paper, we  pushed further the
study of solving  equations over finite fields by deriving  all the
solutions in $\GF{p^n}$ to the generic affine equation $\sum_{i=0}^t
\alpha_i X^{p^i}=a$ where $\alpha_i\in\GF{p}$. This allows us to
provide a generalization of the results obtained in the two previous
articles  \cite{MKCL2020} and \cite{MKCL2019}.


\begin{thebibliography}{10}

\bibitem{BSS1999}
\newblock I. Blake, G. Seroussi and N. Smart.
\newblock Elliptic Curves in Cryptography.
Number 265 in London Mathematical Society Lecture Note Series.
Cambridge University Press, 1999.

\bibitem{camion1983}
\newblock P. Camion.
\newblock Improving an algorithm for factoring polynomials over a finite
field and constructing large irreducible polynomials.
\newblock {\em IEEE Trans. Info. Th.}, 29, 378-385, 1983.

\bibitem{CZ1981}
\newblock D. Cantor and H. Zassenhaus.
\newblock A new algorithm for factoring polynomials over finite
fields.
\newblock {\em Math. Comp.}, 36, 587 -- 592, 1981.

\bibitem{CarletBook}
 C.~Carlet. Boolean Functions for Cryptography and Error Correcting Codes.  Chapter of the monography {\em Boolean Models and Methods in Mathematics, Computer Science, and Engineering}, Y.~Crama and P.~Hammer eds,  Cambridge University Press, pp. 257--397, 2010.

\bibitem{CarletBook1}
C.~Carlet. Vectorial Boolean Functions for Cryptography. Chapter of
the monography {\em Boolean Models and Methods in Mathematics,
Computer Science, and Engineering}, Y.~Crama and P.~Hammer eds,
Cambridge University Press, pp.  398 --469, 2010.


\bibitem{C20}
\newblock B. Csajb\'{o}k.
\newblock Scalar $q-$subresultants and Dickson matrices.
\newblock {\em Journal of Algebra}, 547: 116 -- 128, 2020.

\bibitem{C2019109}
\newblock B. Csajb\'{o}k, G. Marino, O. Polverino and F. Zullo.
\newblock A characterization of linearized polynomials with maximum kernel.
\newblock {\em Finite Fields and Their Applications}, 56:109 -- 130, 2019.


\bibitem{KaltofenShoup1998}
\newblock E. Kaltofen and V. Shoup.
\newblock Subquadratic-time factoring of polynomials over finite fields.
\newblock {\em Math. Comput.}, 67(223)
1179 -- 1197, 1998.

\bibitem{KedlayaUmans08}
\newblock K. Kedlaya and C. Umans.
\newblock Fast modular composition in any
characteristic.
\newblock {\em in: Proceedings of the 49th Annual IEEE Symposium on
Foundations of Computer Science (FOCS)}, 146 -- 155, 2008.


\bibitem{KCM20}
\newblock K. H. Kim, J. Choe and S. Mesnager.
\newblock Solving $X^{q+1}+X+a=0$ over Finite Fields.
\newblock https://arxiv.org/abs/1912.12648, 2019.

\bibitem{KM20}
\newblock K.H. Kim and S. Mesnager.
\newblock Solving $x^{2^k+1}+x+a=0$ in $\mathbb {F}_{2^n}$ with $\text{gcd}(n,k)=1$.
\newblock {\em Finite Fields and Their Applications}, Vol. 63, 2020.
https://doi.org/10.1016/j.ffa.2019.101630

\bibitem{Lidl1997}
R. Lidl and H. Niederreiter. Finite Fields, volume 20 of
Encyclopedia of Mathematics and its Applications, Cambridge
University Press, Cambridge, second edition, 1997.

\bibitem{MCGUIRE201968}
G. McGuire and J. Sheekey.
\newblock A characterization of the number of roots of linearized and
  projective polynomials in the field of coefficients.
\newblock {\em Finite Fields and Their Applications}, 57:68 -- 91, 2019.

\bibitem{McGuireMueller}
G.  McGuire and D. Mueller.
Some results on linearized trinomials that split completely.
\newblock {\em arXiv:1905.11755. Proceedings of Fq14}.

\bibitem{MKCL2020}
\newblock S. Mesnager, K.H. Kim,  J.H. Choe and D.N. Lee.
\newblock Solving some affine equations over finite fields.
\newblock Cryptology ePrint Archive 2020/160, 2020. To appear in journal FFA.

\bibitem{MKCL2019}
\newblock S. Mesnager, K.H. Kim,  J.H. Choe, D.N. Lee and  D.S. Go.
\newblock Solving
$x+x^{2^l}+\cdots+x^{2^{ml}}=a$ over $\mathbb{F}_{2^n}$.
\newblock {\em Cryptography and
Communications}, 12(4): 809--817, 2020.

\bibitem{MKCT20}
\newblock S. Mesnager, K.H. Kim, J. Choe and C. Tang.
\newblock On the Menezes-Teske-Weng's
conjecture.
\newblock {\em Cryptography and Communications}, 12(1): 19 -- 27, 2020.

\bibitem{MKJ20}
\newblock S. Mesnager, K.H. Kim and  M.S. Jo.
\newblock On the number of the rational zeros
of linearized polynomials and the second-order nonlinearity of cubic
Boolean functions.
\newblock {\em Cryptography and Communications}, 12(4): 659--674, 2020.
\bibitem{PF19}
\newblock O. Polverino and F. Zullo.
\newblock On the number of roots of some linearized polynomials.
\newblock arXiv:1909.00802, 2019.

\bibitem{Shoup1990}
\newblock V. Shoup.
\newblock On the deterministic complexity of factoring polynomials over
finite fields.
\newblock {\em Information Processing Letters}, 38, 39 -- 42, 1991.

\bibitem{Shparlinski1992}
\newblock I.E. Shparlinski.
\newblock Computational Problems in Finite Fields.
\newblock Kluwer Academic Publishers, 1992.

\bibitem{WU201379}
B. Wu and Z. Liu.
\newblock Linearized polynomials over finite fields revisited.
\newblock {\em Finite Fields and Their Applications}, 22:79 -- 100, 2013.

\bibitem{Z19}
\newblock C. Zanella.
\newblock A condition for scattered linearized polynomials involving Dickson matrices.
\newblock {\em Journal of Geometry}, 110.3:50, 2019.
\end{thebibliography}
\end{document}